\documentclass[12pt,english]{article}
\usepackage[T1]{fontenc}
\usepackage[latin9]{inputenc}
\usepackage{geometry}
\usepackage{babel}
\usepackage{amsthm}
\usepackage{amsmath}
\usepackage{amssymb}
\usepackage{esint}
\usepackage[unicode=true,pdfusetitle,
 bookmarks=false,
 breaklinks=false,pdfborder={0 0 1},backref=false,colorlinks=false]
 {hyperref}
\hypersetup{
 pdfstartview={FitH}}

\makeatletter
\numberwithin{equation}{section}
\numberwithin{figure}{section}
\theoremstyle{plain}
\newtheorem{thm}{\protect\theoremname}
  \theoremstyle{plain}
  \newtheorem{prop}[thm]{\protect\propositionname}
  \theoremstyle{definition}
  \newtheorem{defn}[thm]{\protect\definitionname}
  \theoremstyle{plain}
  \newtheorem{cor}[thm]{\protect\corollaryname}
  \theoremstyle{plain}
  \newtheorem{lem}[thm]{\protect\lemmaname}

\date{}

\makeatother

  \providecommand{\corollaryname}{Corollary}
  \providecommand{\definitionname}{Definition}
  \providecommand{\lemmaname}{Lemma}
  \providecommand{\propositionname}{Proposition}
\providecommand{\theoremname}{Theorem}

\begin{document}

\title{On the spectral and homological dimension of $\kappa$-Minkowski
space}

\author{Marco Matassa%
\thanks{SISSA, Via Bonomea 265, I-34136 Trieste, Italy\textit{. E-mail address}:
marco.matassa@sissa.it%
}}
\maketitle
\begin{abstract}
We extend the construction of a spectral triple for $\kappa$-Minkowski
space, previously given for the two-dimensional case, to the general
$n$-dimensional case. This takes into account the modular group naturally
arising from the symmetries of the geometry, and requires the use
of notions that have been recently developed in the frameworks of
twisted and modular spectral triples. First we compute the spectral
dimension, using an appropriate weight, and show that in general it
coincides with the classical one. We also study the classical limit
and the analytic continuation of the associated zeta function. Then
we compare this notion of dimension with the one coming from homology.
To this end, we compute the twisted Hochschild dimension of the universal
enveloping algebra underlying $\kappa$-Minkowski space. The result
is that twisting avoids the dimension drop, similarly to other examples
coming from quantum groups. In particular, the simplest such twist
is given by the inverse of the modular group mentioned above.
\end{abstract}

\section{Introduction}

The \textit{$\kappa$-Poincaré} Hopf algebra \cite{luk1,luk2} and
the related non-commutative spacetime \textit{$\kappa$-Minkowski}
\cite{maj} are interesting toy models to study features arising from
quantum gravity. They have been much investigated in the Hopf algebraic
framework, but so far there have been only a few studies \cite{dandrea-kappa,kappa1,kappa2}
from the point of view of (Euclidean) non-commutative geometry in
the sense of Alain Connes \cite{connes}. Recently we argued in \cite{modulark}
that a natural starting point for such a construction is a KMS weight
which is invariant under the $\kappa$-Poincaré symmetries.

This introduces difficulties for the construction of a spectral triple,
since the necessary modifications to the axioms, due to the presence
of the modular group associated to a KMS weight, are not yet well
understood. In the literature two different approaches have been advocated
to deal, at least partly, with these issues: the first one is that
of \textit{twisted spectral triples} \cite{type III}, where one requires
the boundedness of a twisted commutator, where the twist is given
by an automorphism of the algebra; the second one is that of \textit{modular
spectral triples} \cite{modular1,modular2,modular3} where one, roughly
speaking, replaces the operator trace with a weight having a non-trivial
modular group. We used ideas from both of these approaches in \cite{modulark},
where we constructed a {}``modular spectral triple'' for the two
dimensional $\kappa$-Minkowski space. The quotes are necessary because,
even though the construction follows the spirit of this approach,
strictly speaking it does not satisfy its axioms (but see also \cite{kaad-senior,kaad}).
Indeed they arise from the study of geometries with a periodic modular
group, while in our case we do not have this periodicity. Another
general source of difficulties comes from the fact that we are dealing
with a non-compact geometry: this is an aspect which has not been
developed too much in this approach to non-commutative geometry (but
see \cite{moyal}), and we feel that it should be important to make
some progress in this respect to make more contact with physics.

After this general discussion let us recall briefly the construction
given in \cite{modulark}. Starting from an algebra $\mathcal{A}$,
naturally associated to the commutation relations of $\kappa$-Minkowski
space, we introduce a Hilbert space $\mathcal{H}$ via the GNS-construction
for the KMS weight $\omega$ we mentioned above. We emphasize that
the main difference with respect to other approaches is in the choice
of this weight, which is motivated by the symmetries. The use of the
twisted commutator turns out to be necessary if one wants to satisfy
a boundedness condition and have a good classical limit. Moreover
these requests, together with some symmetry conditions, single out
a unique operator $D$ and unique twist $\sigma$ such that the twisted
commutator is bounded.

However the triple $(\mathcal{A},\mathcal{H},D)$ is not finitely
summable, an outcome which is hinted by the mismatch in the modular
properties of the weight $\omega$ and the non-commutative integral.
We will repeat and emphasize this argument later in the paper, so
we do not elaborate further here. This mismatch also hints at the
possibility that, by choosing an appropriate weight in the sense of
modular spectral triples, we can obtain a finite spectral dimension.
This is indeed the case and we find that, in this setting, the spectral
dimension is finite and coincides with the classical one. Moreover,
by computing the residue at the spectral dimension of the corresponding
zeta function, we recover the weight $\omega$ up to a constant.

Having summarized the construction, which is essential for our investigations
here, we now come to the contents of the paper. First we give an extension
of the results mentioned above to the $n$-dimensional case. This
turns out to be fairly easy, and in doing so we will skip most the
details of the construction, which can be easily filled using the
detailed arguments provided in \cite{modulark}. Secondly we provide
further evidence that, although this construction is still not well
understood as part of a general framework, it should be relevant for
the description of the non-commutative geometry of $\kappa$-Minkowski
space. Using the same ingredients of the two-dimensional case, we
find that the spectral dimension according to our definition is in
general equal to the classical one. Moreover, by computing the residue
at the spectral dimension of the associated zeta function, we recover
the weight $\omega$ as in the two-dimensional case. These results
confirm the intuition that, in passing from the two-dimensional case
to the general one, little changes. Next we analyze some properties
of the zeta function that we introduced. We show that, by taking the
limit of the deformation parameter to zero, it reduces as it should
to the classical setting. Also, as in the commutative setting, this
zeta function can be analytically continued to a meromorphic function
on the complex plane, with only simple poles. The poles of the commutative
case still remain, but additional ones appear due to the presence
of the deformation parameter. The significance of these poles remains
to be investigated.

Another important issue we analyze is the homological dimension of
this geometry. In the framework of non-commutative geometry this notion
is given by the dimension of the Hochschild homology, which in the
commutative case coincides with the spectral dimension. However in
many examples, coming in particular from quantum groups, one finds
that the homological dimension is lower that the spectral dimension,
a phenomenon known as \textit{dimension drop}. In many cases it is
possible to avoid this drop by introducing a twist in the homology
theory, as seen for example in \cite{SLq(2),podles-hom}. Here we
compute the \textit{twisted Hochschild homology} \cite{dual-comp}
of the universal enveloping algebra associated to $\kappa$-Minkowski
space. Similarly to the examples we mentioned above, we show that
the dimension drop occurs at the level of Hochschild homology, but
can be avoided by introducing a twist. More interestingly, the simplest
twist which avoids the drop is the inverse of the modular group of
the weight $\omega$, while the other possible twists are given by
its positive powers. This should be compared to the case of \cite{SLq(2),podles-hom}
and other examples, where the twist is the inverse of the modular
group of the Haar state, and therefore seems to be a general feature
of these non-commutative geometries.

\section{The spectral triple}

\subsection{The $\kappa$-Poincaré and $\kappa$-Minkowski algebras}

In this subsection we summarize the algebraic properties of the $\kappa$-Poincaré
and $\kappa$-Minkowski Hopf algebras, which we denote respectively
by $\mathcal{P}_{\kappa}$ and $\mathcal{M}_{\kappa}$. Actually we
can restrict our attention to the translations sector of the $\kappa$-Poincaré
algebra, which is all that is needed to define the $\kappa$-Minkowski
space, while the complete algebra can be obtained by the bicrossproduct
contruction \cite{maj}. The construction in an arbitrary number of
dimensions was given in \cite{kpoincare-dim}. We have the generators
$P_{\mu}$, with $\mu$ ranging from $0$ to $n-1$, which satisfy
$[P_{\mu},P_{\nu}]=0$. The rest of the Hopf algebra structure is
contained in the coproduct $\Delta:\mathcal{P}_{\kappa}\to\mathcal{P}_{\kappa}\otimes\mathcal{P}_{\kappa}$,
the counit $\varepsilon:\mathcal{P}_{\kappa}\to\mathbb{C}$ and antipode
$S:\mathcal{P}_{\kappa}\to\mathcal{P}_{\kappa}$ which are given by
\[
\begin{split} & \Delta(P_{0})=P_{0}\otimes1+1\otimes P_{0}\ ,\quad\Delta(P_{j})=P_{j}\otimes1+e^{-P_{0}/\kappa}\otimes P_{j}\ ,\\
 & \varepsilon(P_{\mu})=0\ ,\quad S(P_{0})=-P_{0}\ ,\quad S(P_{j})=-e^{P_{0}/\kappa}P_{j}\ .
\end{split}
\]
We adopt the usual general relativistic convention of greek indices
going from $0$ to $n-1$, while latin indices go from $1$ to $n-1$.
From the defining relations we see that the translations sector is
an Hopf subalgebra of $\mathcal{P}_{\kappa}$, and we denote it by
$\mathcal{T}_{\kappa}$. We define the $\kappa$-Minkowski space $\mathcal{M}_{\kappa}$
as the dual Hopf algebra to this subalgebra \cite{maj}. If we denote
the pairing by $\langle\cdot,\cdot\rangle:\mathcal{T}_{\kappa}\times\mathcal{M}_{\kappa}\to\mathbb{C}$,
then the structure of $\mathcal{M}_{\kappa}$ is determined by the
duality relations 
\[
\begin{split}\langle t,xy\rangle & =\langle t^{(1)},x\rangle\langle t^{(2)},y\rangle\ ,\\
\langle ts,x\rangle & =\langle t,x^{(1)}\rangle\langle s,x^{(2)}\rangle\ .
\end{split}
\]
Here we have $t,s\in\mathcal{T}_{\kappa}$, $x,y\in\mathcal{M}_{\kappa}$
and we use the Sweedler notation for the coproduct
\[
\Delta x=\sum_{i}x_{(i)}^{(1)}\otimes x_{(i)}^{(2)}=x^{(1)}\otimes x^{(2)}\ .
\]
From the pairing we deduce that $\mathcal{M}_{\kappa}$ is non-commutative,
since $\mathcal{T}_{\kappa}$ is not cocommutative, that is the coproduct
in $\mathcal{T}_{\kappa}$ is not trivial. On the other hand, since
$\mathcal{T}_{\kappa}$ is commutative we have that $\mathcal{M}_{\kappa}$
is cocommutative. The algebraic relations for the $\kappa$-Minkowski
Hopf algebra $\mathcal{M}_{\kappa}$ are

\[
[X^{0},X^{j}]=-\kappa^{-1}X^{j}\ ,\qquad\Delta X^{\mu}=X^{\mu}\otimes1+1\otimes X^{\mu}\ .
\]

This concludes the usual presentation of the $\kappa$-Poincaré and
$\kappa$-Minkowski Hopf algebras. In this paper we are going to use
a slightly different presentation, as done in \cite{sitarz}, to avoid
issues of convergence in the algebra. Instead of considering the exponential
$e^{-P_{0}/\kappa}$ as a power series in $P_{0}$, we consider it
as an invertible element $\mathcal{E}$ and rewrite the defining relations
as
\[
\begin{split} & [P_{\mu},P_{\nu}]=0\ ,\quad[P_{\mu},\mathcal{E}]=0\ ,\\
 & \Delta(P_{0})=P_{0}\otimes1+1\otimes P_{0}\ ,\quad\Delta(P_{j})=P_{j}\otimes1+\mathcal{E}\otimes P_{j}\ ,\quad\Delta(\mathcal{E})=\mathcal{E}\otimes\mathcal{E}\ ,\\
 & \varepsilon(P_{\mu})=0\ ,\quad\varepsilon(\mathcal{E})=1\ ,\\
 & S(P_{0})=-P_{0}\ ,\quad S(P_{j})=-\mathcal{E}^{-1}P_{j}\ ,\quad S(\mathcal{E})=\mathcal{E}^{-1}\ .
\end{split}
\]
In this form, the role of the subalgebra $\mathcal{T}_{\kappa}$ is
played by the one generated by $P_{\mu}$ and $\mathcal{E}$, which
we call the \textit{extended momentum algebra} and denote again by
$\mathcal{T}_{\kappa}$. An appropriate pairing defining $\kappa$-Minkowski
space can be easily written in terms of these generators. More importantly
it can be made into a Hopf $*$-algebra by defining the involution
as $P_{\mu}^{*}=P_{\mu}$ and $\mathcal{E}^{*}=\mathcal{E}$.

In the following we are going to consider the case of Euclidean signature
and so, strictly speaking, we should refer to the Euclidean counterpart
of the $\kappa$-Poincaré algebra, which is known as the \textit{quantum
Euclidean group}. However the boost generator $N$ is not going to
play a central role in our discussion, which is going to be based
on the extended momentum algebra, and therefore most of our relations
do not depend on the signature. Henceforth we only make reference
to the $\kappa$-Poincaré algebra and make some remarks when needed.

One more remark on the notation: we are going to write all formulae
in terms of the parameter $\lambda:=\kappa^{-1}$, instead of $\kappa$.
The motivation comes from the fact that the Poincaré algebra is obtained
in the {}``classical limit'' $\lambda\to0$, in a similar fashion
to the classical limit $\hbar\to0$ of quantum mechanics. This makes
more transparent checking that some formulae reduce, in this limit,
to their respective undeformed counterparts.

\subsection{The algebraic construction}

We begin by generalizing to $n$ dimensions the construction of the
$*$-algebra given in \cite{sitarz}. We will skip most of the computations,
since they are completely analogous to the two-dimensional case, but
we will provide some details regarding the modular aspects of the
construction.

The underlying algebra of the $n$-dimensional $\kappa$-Minkowski
space is the enveloping algebra of the Lie algebra with generators
$ix^{0}$ and $ix^{k}$ (with $k=1,\ldots,n-1$), fullfilling the
commutation relations $[x^{0},x^{k}]=i\lambda x^{k}$. It has a faithful
$n\times n$ matrix representation $\varphi$ given by
\[
\varphi(ix^{0})=\left(\begin{array}{ccc}
-\lambda & \cdots & 0\\
\vdots & \ddots & 0\\
0 & \cdots & 0
\end{array}\right)\ ,\qquad\varphi(ix^{k})=\left(\begin{array}{ccccc}
0 & \cdots & 1 & \cdots & 0\\
\vdots & \ddots & \ddots & \ddots & 0\\
0 & \cdots & 0 & \cdots & 0
\end{array}\right)\ .
\]
The matrix $\varphi(ix^{k})$ has non-zero values only in the $(k+1)$-th
column. An element of the associated group $G$ can be presented in
the form
\begin{equation}
S(a)=\left(\begin{array}{cccc}
e^{-\lambda a_{0}} & a_{1} & \cdots & a_{n-1}\\
0 & 1 & 0 & 0\\
\vdots & \ddots & \ddots & \vdots\\
0 & 0 & \cdots & 1
\end{array}\right)\ .\label{eq:param}
\end{equation}
Here we use the notation $a=(a_{0},\vec{a})$, where $\vec{a}=(a_{1},\cdots,a_{n-1})$.
The group operations written in the $(a_{0},\cdots,a_{n-1})$ coordinates
are given by
\begin{equation}
S(a)S(b)=S(a_{0}+b_{0},\vec{a}+e^{-\lambda a_{0}}\vec{b})\ ,\qquad S(a)^{-1}=S(-a_{0},-e^{\lambda a_{0}}\vec{a})\ .\label{eq:g-ops}
\end{equation}

\begin{prop}
The left and right invariant measures on $G$ are given respectively
by $d\mu_{L}(a)=e^{\lambda(n-1)a_{0}}d^{n}a$ and $d\mu_{R}(a)=d^{n}a$,
where $d^{n}a$ is the Lebesgue measure on $\mathbb{R}^{n}$.\end{prop}
\begin{proof}
We do the computation for the left invariant case. Using the $(a_{0},\cdots,a_{n-1})$
coordinates and the group operations given in (\ref{eq:g-ops}) we
easily find
\[
\begin{split}\int f(a\cdot b)d\mu_{L}(b) & =\int f(a_{0}+b_{0},\vec{a}+e^{-\lambda a_{0}}\vec{b})e^{\lambda(n-1)b_{0}}d^{n}b\\
 & =\int f(a_{0}+b_{0},\vec{a}+\vec{b})e^{\lambda(n-1)(a_{0}+b_{0})}d^{n}b\\
 & =\int f(b_{0},\vec{b})e^{\lambda(n-1)b_{0}}d^{n}b=\int f(b)d\mu_{L}(b)\ .
\end{split}
\]
The right invariant case is treated similarly.
\end{proof}
We have that $G$ is not a unimodular group, with the modular function
$e^{-\lambda(n-1)a_{0}}$ playing a central role in the following.
We consider the convolution algebra of $G$ with respect to the right
invariant measure, and we identify functions on $G$ with functions
on $\mathbb{R}^{n}$ by the parametrization (\ref{eq:param}). The
convolution algebra is an involutive Banach algebra consisting of
integrable functions on $\mathbb{R}^{n}$ with product $\hat{\star}$
and involution $\hat{*}$ given by
\[
\begin{split}(f\hat{\star}g)(a) & =\int f(a_{0}-a_{0}^{\prime},\vec{a}-e^{-\lambda(a_{0}-a_{0}^{\prime})}\vec{a}^{\prime})g(a_{0}^{\prime},\vec{a}^{\prime})d^{n}a^{\prime}\ ,\\
f^{\hat{*}}(a) & =e^{\lambda(n-1)a_{0}}\overline{f}(-a_{0},-e^{\lambda a_{0}}\vec{a})\ .
\end{split}
\]
We pass from momentum space to configuration space via the Fourier
transform. The star-product and involution are defined, in terms of
the convolution algebra operations, as 
\[
f\star g:=\mathcal{F}^{-1}\left(\mathcal{F}(f)\hat{\star}\mathcal{F}(g)\right)\ ,\qquad f^{*}:=\mathcal{F}^{-1}\left(\mathcal{F}(f)^{\hat{*}}\right)\ .
\]
These formulae are written for clarity using the unitary convention
for the Fourier transform, but in the following we will use the physicists
convention with the $(2\pi)^{n}$ in momentum space. We restrict our
attention to the following space of functions.
\begin{defn}
Denote by $\mathcal{S}_{c}$ the space of Schwartz functions on $\mathbb{R}^{n}$
with compact support in the first variable, that is for $f\in\mathcal{S}_{c}$
we have $\mbox{supp}(f)\subseteq K\times\mathbb{R}^{n-1}$ for some
compact $K\subset\mathbb{R}$. We define $\mathcal{A}=\mathcal{F}(\mathcal{S}_{c})$,
where $\mathcal{F}$ is the Fourier transform on $\mathbb{R}^{n}$.
\end{defn}
On this space we can safely perform all the operations we need. The
next proposition shows that $\mathcal{A}$ is a $*$-algebra and gives
explicit formulae for the star product and the involution. We use
the notation $x=(x_{0},\vec{x})$ and $\vec{x}=(x_{1},\ldots,x_{n-1})$
for the coordinates on $\mathcal{A}$.
\begin{prop}
For $f,g\in\mathcal{A}$ we have
\[
\begin{split}(f\star g)(x) & =\int e^{ip_{0}x_{0}}(\mathcal{F}_{0}f)(p_{0},\vec{x})g(x_{0},e^{-\lambda p_{0}}\vec{x})\frac{dp_{0}}{2\pi}\ ,\\
f^{*}(x) & =\int e^{ip_{0}x_{0}}(\mathcal{F}_{0}\overline{f})(p_{0},e^{-\lambda p_{0}}\vec{x})\frac{dp_{0}}{2\pi}\ .
\end{split}
\]
We have that $f\star g\in\mathcal{A}$ and $f^{*}\in\mathcal{A}$,
so that $\mathcal{A}$ is a $*$-algebra.\end{prop}
\begin{proof}
Let us show how this works for the involution. From the definitions
we get
\[
\begin{split}f^{*}(x) & =\mathcal{F}^{-1}(\mathcal{F}(f)^{\hat{*}})(x)=\int e^{ipx}(\mathcal{F}(f)^{\hat{*}})(p)\frac{d^{n}p}{(2\pi)^{n}}\\
 & =\int e^{ipx}e^{(n-1)\lambda p_{0}}(\overline{\mathcal{F}f})(-p_{0},-e^{\lambda p_{0}}\vec{p})\frac{d^{n}p}{(2\pi)^{n}}\ .
\end{split}
\]
Now using the change of variables $\vec{p}\to e^{-\lambda p_{0}}\vec{p}$
we find
\[
f^{*}(x)=\int e^{ip_{0}x_{0}}e^{ie^{-\lambda p_{0}}\vec{p}\cdot\vec{x}}(\overline{\mathcal{F}f})(-p)\frac{d^{n}p}{(2\pi)^{n}}=\int e^{ip_{0}x_{0}}e^{ie^{-\lambda p_{0}}\vec{p}\cdot\vec{x}}(\mathcal{F}\bar{f})(p)\frac{d^{n}p}{(2\pi)^{n}}\ .
\]
Finally performing the Fourier transform in the $\vec{p}$ variables
we find
\[
f^{*}(x)=\int e^{ip_{0}x_{0}}(\mathcal{F}_{0}\bar{f})(p_{0},e^{-\lambda p_{0}}\vec{x})\frac{dp_{0}}{2\pi}
\]
The product can be computed in a similar way. For more details see
\cite{sitarz}.
\end{proof}
A nice property of this algebra is that it comes naturally with an
action of the $\kappa$-Poincaré algebra $\mathcal{P}_{\kappa}$ on
it. In \cite{sitarz} it was proven that $\mathcal{A}$ is a left
$\mathcal{P}_{\kappa}$-module $*$-algebra, which means that the
action of the $\kappa$-Poincaré symmetries on $\mathcal{A}$ preserves
the Hopf algebraic structure. In particular the action of the translations
sector is elementary, with $(P_{\mu}\triangleright f)(x)=-i(\partial_{\mu}f)(x)$
and $(\mathcal{E}\triangleright f)(x)=f(x_{0}+i\lambda,\vec{x})$.
This remains true for the $n$-dimensional case.

Now we can introduce a Hilbert space by using the GNS-construction
for $\mathcal{A}$, after the choice of some weight $\omega$. There
is a natural choice which respects the symmetries of the $\kappa$-Poincaré
Hopf algebra, see \cite{sitarz,flavio}: it is simply given by the
integral of a function $f\in\mathcal{A}$ with respect to the Lebesgue
measure over $\mathbb{R}^{n}$, and we denote it by $\omega$. Contrarily
to the commutative case it does not satisfy the trace property.
\begin{prop}
For $f,g\in\mathcal{A}$ we have the twisted trace property
\[
\int(f\star g)(x)d^{n}x=\int(\sigma^{n-1}(g)\star f)(x)d^{n}x\ ,
\]
where we define $\sigma(g)(x):=g(x_{0}+i\lambda,\vec{x})$.\end{prop}
\begin{proof}
First we use the change of variables $\vec{x}\to e^{\lambda p_{0}}\vec{x}$
and obtain
\[
\begin{split}\int(f\star g)(x)d^{n}x & =\int\int e^{ip_{0}x_{0}}(\mathcal{F}_{0}f)(p_{0},\vec{x})g(x_{0},e^{-\lambda p_{0}}\vec{x})\frac{dp_{0}}{2\pi}d^{n}x\\
 & =\int\int e^{ip_{0}x_{0}}e^{(n-1)\lambda p_{0}}(\mathcal{F}_{0}f)(p_{0},e^{\lambda p_{0}}\vec{x})g(x_{0},\vec{x})d^{n}x\frac{dp_{0}}{2\pi}\ .
\end{split}
\]
Now, using the analiticity of the functions of $\mathcal{A}$ in the
first variable, we can shift $x_{0}\to x_{0}+i(n-1)\lambda$ to obtain
the action of $\sigma^{n-1}$ on $g$, that is 
\[
\begin{split}\int(f\star g)(x)d^{n}x & =\int\int e^{ip_{0}x_{0}}(\mathcal{F}_{0}f)(p_{0},e^{\lambda p_{0}}\vec{x})g(x_{0}+i(n-1)\lambda,\vec{x})d^{n}x\frac{dp_{0}}{2\pi}\\
 & =\int\int e^{ip_{0}x_{0}}(\mathcal{F}_{0}f)(p_{0},e^{\lambda p_{0}}\vec{x})\sigma^{n-1}(g)(x_{0},\vec{x})d^{n}x\frac{dp_{0}}{2\pi}\ .
\end{split}
\]
It only remains to rewrite this expression in terms of the $\star$-product.
Writing explicitely the Fourier transform $\mathcal{F}_{0}f$ we have
\[
\int(f\star g)(x)d^{n}x=\int\int e^{ip_{0}x_{0}}\int e^{-ip_{0}y_{0}}f(y_{0},e^{\lambda p_{0}}\vec{x})\sigma^{n-1}(g)(x_{0},\vec{x})dy_{0}d^{n}x\frac{dp_{0}}{2\pi}\ .
\]
We need to do some rearranging: change $p_{0}\to-p_{0}$, relabel
$y_{0}\leftrightarrow x_{0}$ and exchange the order of the $x_{0}$
and $y_{0}$ integral. The result of these operations is
\[
\int(f\star g)(x)d^{n}x=\int\int e^{ip_{0}x_{0}}f(x_{0},e^{-\lambda p_{0}}\vec{x})\int e^{-ip_{0}y_{0}}\sigma^{n-1}(g)(y_{0},\vec{x})dy_{0}d^{n}x\frac{dp_{0}}{2\pi}\ .
\]
But now the last integral is just the Fourier transform of $\sigma^{n-1}(g)$
in the $y_{0}$ variable, so
\[
\int(f\star g)(x)d^{n}x=\int\int e^{ip_{0}x_{0}}(\mathcal{F}_{0}\sigma^{n-1}(g))(p_{0},\vec{x})f(x_{0},e^{-\lambda p_{0}}\vec{x})\frac{dp_{0}}{2\pi}d^{n}x\ .
\]
Finally we observe that the right hand side is just the integral of
the function $(\sigma^{n-1}(g)\star f)(x)$, which proves the result.
\end{proof}
As observed in \cite{modulark} this property can be rephrased as
a KMS condition for $\omega$.
\begin{prop}
The weight $\omega$ satisfies the KMS condition with respect to the
modular group $\sigma^{\omega}$, defined by $(\sigma_{t}^{\omega}f)(x_{0},\vec{x}):=f(x_{0}-t(n-1)\lambda,\vec{x})$.
The associated modular operator is $\Delta_{\omega}=e^{-(n-1)\lambda P_{0}}$,
where $P_{0}=-i\partial_{0}$.
\end{prop}
On the Hilbert space $\mathcal{H}$, obtained by the GNS-construction
for $\omega$, the algebra $\mathcal{A}$ acts via left multiplication,
that is $\pi(f)\psi:=f\star\psi$. In the following we omit the representation
symbol $\pi$ and just write $f$ for the operator of left multiplication
by this function.

It is important to point out that the Hilbert space $\mathcal{H}$
is not $L^{2}(\mathbb{R}^{n})$. On the other hand, using the fact
that $\mathcal{A}$ is dense in both Hilbert spaces, one can easily
find a unitary operator between the two. One can also find, using
this unitary operator, the Schwartz kernel of a certain class of operators
which will be of interest to us in the following. These results are
the content of the next proposition, for details about the derivation
see \cite{modulark}.
\begin{prop}
\label{prop:kernel}The Hilbert space $\mathcal{H}$ obtained by the
GNS-construction for $\omega$ is unitarily equivalent to $L^{2}(\mathbb{R}^{n})$,
via the unitary operator given by
\[
(Uf)(x)=\int e^{ip_{0}x_{0}}(\mathcal{F}_{0}\overline{f})(p_{0},e^{-\lambda p_{0}}\vec{x})\frac{dp_{0}}{2\pi}\ .
\]
Consider now the operator $U\pi(f)g(P)U^{-1}$ acting on $L^{2}(\mathbb{R}^{n})$,
where $f\in\mathcal{A}$ and $P_{\mu}=-i\partial_{\mu}$. Then its
Schwartz kernel is given by
\[
K(x,y)=\int e^{ip(x-y)}(Uf)(x_{0},e^{\lambda p_{0}}\vec{x})g(p_{0},e^{-\lambda p_{0}}\vec{p})\frac{d^{n}p}{(2\pi)^{n}}\ .
\]

\end{prop}

\subsection{Dirac operator and differential calculus}

The next step is the introduction of a self-adjoint operator $D$
satisfying certain conditions, the so-called Dirac operator. From
the analysis given in \cite{modulark} we know that to obtain a boundedness
condition for $D$ we need to use a twisted commutator \cite{type III}.
This amounts to introducing an automorphism $\sigma$ of the algebra
$\mathcal{A}$, the twist. Then for each $f\in\mathcal{A}$ the operator
$[D,f]_{\sigma}=Df-\sigma(f)D$ should be bounded. One needs to find
what are the possible choices for $D$ and the automorphism $\sigma$
such that this condition is fulfilled. We consider some additional
assumptions which are related to the symmetries and the classical
limit, which we state precisely below. The analysis for the general
$n$-dimensional case is essentially identical to the two-dimensional
case, so we skip the computations and refer to \cite{modulark} for
details.

As in the classical case we enlarge the Hilbert space to accomodate
for spinors. Therefore we consider $\mathcal{H}=\mathcal{H}_{r}\otimes\mathbb{C}^{[n/2]}$,
where $\mathcal{H}_{r}$ is the Hilbert space previously introduced.
Here $[n/2]$ is the dimension of the spinor bundle on $\mathbb{R}^{n}$,
and we use the notation $\Gamma^{\mu}$ for the matrix representation
of the Clifford algebra, which satisfy $\{\Gamma^{\mu},\Gamma^{\nu}\}=2\delta^{\mu\nu}$.
Then we can write $D$ in the form $D=\Gamma^{\mu}\hat{D}_{\mu}$,
where $\hat{D}_{\mu}$ are self-adjoint operators on $\mathcal{H}_{r}$.

Now we state our assumptions for the Dirac operator $D$ and the automorphism
$\sigma$. We denote by $\rho$ the map from the extended momentum
algebra $\mathcal{T}_{\kappa}$ to (possibly unbounded) operators
on $\mathcal{H}$, see \cite{modulark}. Since $D$ should be determined
by the symmetries, we assume that $\hat{D}_{\mu}=\rho(D_{\mu})$ for
some $D_{\mu}\in\mathcal{T}_{\kappa}$, which is basically the requirement
of equivariance. Similarly we assume that $\sigma$ is given by $\sigma(f)=\sigma\triangleright f$
for some $\sigma\in\mathcal{T}_{\kappa}$, which to be an automorphism
must have a coproduct of the form $\Delta(\sigma)=\sigma\otimes\sigma$.
Since the parameter $\lambda$ is a \textit{physical quantity} of
the model, which has the dimension of a length, the Dirac operator
must have the dimension of an inverse length. Moreover we require
that $D$ reduces to the classical Dirac operator in the limit $\lambda\to0$,
by which we mean that for all $\psi\in\mathcal{A}$ we should have
$\lim\hat{D}_{\mu}\psi=\hat{P}_{\mu}\psi$.
\begin{prop}
Under the assumptions given above, we have that there is a unique
operator $D$ and a unique automorphism $\sigma$ such that $[D,f]_{\sigma}$
is bounded for every $f\in\mathcal{A}$. They are given by $D=\Gamma^{\mu}D_{\mu}$,
with $D_{0}=\lambda^{-1}(1-e^{-\lambda P_{0}})$ and $D_{j}=P_{j}$,
while $\sigma=e^{-\lambda P_{0}}$.
\end{prop}
Notice that formally for $\lambda\to0$ we obtain the usual Dirac
operator on $\mathbb{R}^{n}$. We have the interesting relation $\Delta_{\omega}D^{2}=C$,
where $\Delta_{\omega}$ is the modular operator of the weight $\omega$
and $C$ is the first Casimir of the $\kappa$-Poincaré algebra, which
is given by
\[
C=\frac{4}{\lambda^{2}}\sinh^{2}\left(\frac{\lambda P_{0}}{2}\right)+\sum_{j=1}^{n-1}e^{\lambda P_{0}}P_{j}^{2}\ .
\]

Now we discuss some aspects of the differential calculus associated
with the operator $D$. For a spectral triple $(\mathcal{A},\mathcal{H},D)$
one defines the $\mathcal{A}$-bimodule $\Omega_{D}^{1}$ of one-forms
as the linear span of operators of the form $a[D,b]$, with $a,b\in\mathcal{A}$.
Then $d(a)=[D,a]$ is a derivation of $\mathcal{A}$ with values in
$\Omega_{D}^{1}$, that is $d(ab)=d(a)b+ad(b)$, which immediately
follows from the properties of the commutator. In the twisted case
this definition must be modified, since we have
\[
[D,ab]_{\sigma}=[D,a]_{\sigma}b+\sigma(a)[D,b]_{\sigma}\ .
\]
The necessary modification is very simple \cite{type III}. One simply
defines $\Omega_{D}^{1}$ to be the linear span of operators of the
form $a[D,b]_{\sigma}$, with the bimodule structure given by $a\cdot[D,b]_{\sigma}\cdot c=\sigma(a)[D,b]_{\sigma}c$.
Then it is obvious that $d_{\sigma}(a)=[D,a]_{\sigma}$ is a derivation
of $\mathcal{A}$ with values in $\Omega_{D}^{1}$.

In the non-compact case, already at the untwisted level, it is not
completely clear how one should generalize this notion. One can replace
the algebra $\mathcal{A}$ with some unitization, as done in \cite{moyal}.
However in this case there is no analogue of the one-form $dx^{\mu}$,
since the function $x^{\mu}$ does not belong to $\mathcal{A}$ or
some unitization of it. Nevertheless, it is clear that in the commutative
case $[D,x^{\mu}]$ extends to a bounded operator, in particular it
is equal to $-i\Gamma^{\mu}$. This is also the case in this non-commutative
setting: indeed notice that for any $f\in\mathcal{A}$ we have the
equality
\[
[D,f]_{\sigma}=\Gamma^{\mu}(D_{\mu}\triangleright f)\ ,
\]
where $D_{0}=\lambda^{-1}(1-e^{-\lambda P_{0}})$ and $D_{j}=P_{j}$.
Now it is easy to see that the twisted commutator $[D,x^{\mu}]_{\sigma}$
extends to a bounded operator and in particular $[D,x^{\mu}]_{\sigma}=-i\Gamma^{\mu}$,
as in the commutative case. Adopting the natural notation $df=[D,f]_{\sigma}$,
we can write $df=dx^{\mu}(iD_{\mu}\triangleright f)$. Then from the
bimodule structure on $\Omega_{D}^{1}$ it then follows that 
\[
df=dx^{\mu}\cdot(iD_{\mu}\triangleright f)=\sigma^{-1}(iD_{\mu}\triangleright f)\cdot dx^{\mu}\ .
\]

In \cite{mink-diff} the introduction of bicovariant differential
calculi on $\kappa$-Minkowski space was investigated. It follows
from our construction that the differential calculus defined by the
operator $D$ is an example of such a bicovariant differential calculus.
We have the relation
\[
x^{\mu}\cdot dx^{\nu}-dx^{\nu}\cdot x^{\mu}=\sigma(x^{\mu})dx^{\nu}-dx^{\nu}x^{\mu}=i\lambda\delta_{0}^{\mu}dx^{\nu}\ .
\]
Therefore in the notation of \cite{mink-diff} we obtain $[x^{\mu},dx^{\nu}]=iA_{\rho}^{\mu\nu}dx^{\rho}$
with $A_{\rho}^{\mu\nu}=\lambda\delta_{0}^{\mu}\delta_{\rho}^{\nu}$.

\section{Spectral dimension}

\subsection{The spectral dimension}

In \cite{modulark} it was shown, for the two-dimensional case, that
the ingredients introduced in the previous section do not give a finitely
summable (twisted) spectral triple. We can try to interpret this result
with the following heuristic argument: suppose we did find a spectral
dimension $n$, coinciding with the classical dimension. Then, from
the general properties of twisted spectral triples, it would follow
that $\varphi(ab)=\varphi\left(\sigma^{n}(b)a\right)$, where $\sigma(a)=e^{-\lambda P_{0}}ae^{\lambda P_{0}}$
and $\varphi$ is the non-commutative integral (defined, for example,
in terms of the Dixmier trace). The weight $\omega$, on the other
hand, satisfies $\omega(f\star g)=\omega(\sigma_{i}^{\omega}(g)\star f)$,
where $\sigma_{i}^{\omega}(a)=e^{-(n-1)\lambda P_{0}}ae^{(n-1)\lambda P_{0}}$.
Therefore we have a mismatch between the modular properties of the
weight $\omega$ and the integral $\varphi$, which shows that we
can not recover the weight $\omega$ from the non-commutative integral.

This argument leaves open the possibility that this could happen if
the spectral dimension were equal to $n-1$, but this is shown not
to be the case by the explicit calculation. In \cite{modulark} it
was argued that one has to use a weight to obtain finite summability,
as in the framework of modular spectral triples \cite{modular1,modular2,modular3}.
The relevant definition for us is the following.
\begin{defn}
Let $(\mathcal{A},\mathcal{H},D)$ be a non-compact modular spectral
triple with weight $\Phi$. We say that it is finitely summable and
call $p$ the \textit{spectral dimension} if the following quantity
exists 
\[
p:=\inf\{s>0:\forall a\in\mathcal{A},a\geq0,\ \Phi\left(a(D^{2}+1)^{-s/2}\right)<\infty\}\ .
\]

\end{defn}
We can choose our weight to be of the form $\Phi(\cdot)=\mathrm{Tr}(\Delta_{\Phi}\cdot)$,
where $\Delta_{\Phi}$ is a positive and invertible operator. We call
it the \textit{modular operator} associated to the weight $\Phi$,
and denote the corresponding modular group by $\sigma^{\Phi}$. As
we discussed above, the mismatch of one power of $e^{-\lambda P_{0}}$
suggests setting $\Delta_{\Phi}=e^{-\lambda P_{0}}$ as the modular
element. It is instructive to consider a slightly more general situation,
which we discuss in the following proposition.
\begin{prop}
\label{prop:spec-dim}Let $\Phi_{t}(\cdot)=\mathrm{Tr}(\Delta_{\Phi}^{t}\cdot)$
be the weight with modular operator $\Delta_{\Phi}^{t}=e^{-t\lambda P_{0}}$.
Then, for any $f\in\mathcal{A}$, we have $\Phi_{t}\left(f(D^{2}+\mu^{2})^{-s/2}\right)<\infty$
if and only if $t>0$ and $s>n-1+t$. In other words, the spectral
dimension exists for $t>0$ and is given by $p=n-1+t$.\end{prop}
\begin{proof}
First of all notice that we have $\Delta_{\Phi}^{t}f=\sigma^{t}(f)\Delta_{\Phi}^{t}$,
so that we can consider without loss of generality the operator $A:=f\Delta_{\Phi}^{t}(D^{2}+\mu^{2})^{-s/2}$.
Using the unitary operator $U$ we can consider $A$ as an operator
on $L^{2}(\mathbb{R}^{n}\otimes\mathbb{C}^{2^{[n/2]}})$ whose symbol,
thanks to Proposition \ref{prop:kernel}, is given by $a(x,\xi):=(Uf)(x_{0},e^{\lambda\xi_{0}}\vec{x})G_{s,t}^{\Delta}(\xi_{0},e^{-\lambda\xi_{0}}\vec{\xi})$,
where we have defined
\[
G_{s,t}^{\Delta}(\xi)=e^{-t\lambda\xi_{0}}\left(\lambda^{-2}\left(1-e^{-\lambda\xi_{0}}\right)^{2}+\vec{\xi}^{2}+\mu^{2}\right)^{-s/2}\ .
\]
To prove that the operator $A$ is trace-class it suffices to show
that its symbol and certain number of its derivatives are integrable
\cite{arsu}. We now show that the symbol $a(x,\xi)$ is integrable.
With a simple change of variables we can factorize the integral as
\[
\int|a(x,\xi)|d^{n}xd^{n}\xi=\int G_{s,t}^{\Delta}(\xi)d^{n}\xi\int|(Uf)(x)|d^{n}x\ .
\]
The integral of $(Uf)(x)$ is clearly finite. We can now perform the
integral in the variables $(\xi_{1},\cdots,\xi_{n-1})$ using the
well known formula 
\[
\int(\xi^{2}+a^{2})^{-z/2}d^{N}\xi=\pi^{N/2}\frac{\Gamma\left(\frac{z-N}{2}\right)}{\Gamma\left(\frac{z}{2}\right)}a^{-(z-N)}\ ,
\]
which is valid for $\mathrm{Re}(z)>N$. Then we have
\[
\int G_{s,t}^{\Delta}(\xi)d^{n}\xi=\pi^{(n-1)/2}\frac{\Gamma\left(\frac{s-(n-1)}{2}\right)}{\Gamma\left(\frac{s}{2}\right)}\int e^{-t\lambda\xi_{0}}\left(\lambda^{-2}\left(1-e^{-\lambda\xi_{0}}\right)^{2}+\mu^{2}\right)^{-\frac{s-(n-1)}{2}}d\xi_{0}\ ,
\]
provided that $s>n-1$. To proceed further we consider the asymptotics
of the integrand
\[
\tilde{I}_{t}(s):=e^{-t\lambda\xi_{0}}\left(\lambda^{-2}\left(1-e^{-\lambda\xi_{0}}\right)^{2}+\mu^{2}\right)^{-\frac{s-(n-1)}{2}}\ .
\]
 For $\xi_{0}\to+\infty$ we have $\tilde{I}_{t}\sim e^{-t\lambda|\xi_{0}|}$,
so it integrable provided that $t>0$, independently of $s$. In the
other regime $\xi_{0}\to-\infty$ we have instead
\[
\tilde{I}_{t}\sim e^{t\lambda|\xi_{0}|}e^{-(s-(n-1))\lambda|\xi_{0}|}=e^{-(s-(n-1)-t)\lambda|\xi_{0}|}\ ,
\]
which is integrable when $s>n-1+t$. It is easy to see that the various
derivatives of the symbol $a(x,\xi)$ are integrable under these conditions.
Finally taking the infimum over $s$ we obtain that the spectral dimension
is $p=n-1+t$.
\end{proof}
We see that by introducing the weight $\Phi_{t}$ we are able to obtain
a finite spectral dimension. But what about the free parameter $t$?
A natural choice is to fix $t=1$, in such a way that the spectral
dimension coincides with the classical dimension $n$. To understand
this ambiguity let us consider for a moment a generic situation, where
we have a weight $\Phi(\cdot)=\mathrm{Tr}(\Delta_{\Phi}\cdot)$ and
a spectral dimension $p$. If we define the non-commutative integral
$\varphi$ in terms of $\Phi$, as we will do later in terms of a
zeta function, then on general grounds we expect to have the property
\[
\varphi(ab)=\varphi\left(\sigma_{i}^{\Phi}\left(\sigma^{p}(b)\right)a\right)\ .
\]
Here $\sigma$ is the twist coming from the twisted commutator and
$\sigma^{\Phi}$ is the modular group of $\Phi$. This is simply what
happens in the twisted setting, except for the presence of the additional
twist $\sigma^{\Phi}$. Let us consider now our specific case: the
twist of the commutator is $\sigma(a)=e^{-\lambda P_{0}}ae^{\lambda P_{0}}$,
the twist of the weight is $\sigma_{i}^{\Phi_{t}}(a)=e^{t\lambda P_{0}}be^{-t\lambda P_{0}}$
and for $t>0$ the spectral dimension is given by $p=n-1+t$. Using
these formulae we can compute
\[
\begin{split}\sigma_{i}^{\Phi_{t}}\left(\sigma^{p}(b)\right)a & =e^{t\lambda P_{0}}e^{-(n-1+t)\lambda P_{0}}be^{(n-1+t)\lambda P_{0}}e^{-t\lambda P_{0}}a\\
 & =e^{-(n-1)\lambda P_{0}}be^{(n-1)\lambda P_{0}}a=\sigma_{i}^{\omega}(b)a\ .
\end{split}
\]
We learn, from this short computation, that this specific combination
allows us to recover the KMS condition for the weight $\omega$. This
provides strength to the argument that recovering the weight $\omega$
from the non-commutative integral provides the right guidance in this
setting. At the same time it shows that this is not enough to fix
the free parameter $t$, since it disappears in the combination $\sigma^{\Phi_{t}}\circ\sigma^{p}$.
We do not know at the moment what kind of condition could select the
value $t=1$ uniquely (apart from recovering the classical dimension,
of course).

\subsection{Poles of the zeta function}

In the following we fix $t=1$. Then the function $\Phi\left(f(D^{2}+\mu^{2})^{-s/2}\right)$
has a singularity at $s=n$, whose nature we now want to investigate,
along with its analytic continuation to the complex plane. The singularities
of this kind of {}``zeta function'' play an important role in the
local index formula of Connes and Moscovici \cite{local-index}. Before
starting the analysis, let us briefly review the commutative case
of $\mathbb{R}^{n}$. The Dirac operator is given by $D=-i\Gamma^{\mu}\partial_{\mu}$,
where $\Gamma^{\mu}$ are the gamma matrices satisfying the relations
$\{\Gamma^{\mu},\Gamma^{\nu}\}=2\delta^{\mu\nu}$ and the dimension
of the spinor bundle is $2^{[n/2]}$. We consider the zeta function
defined by
\[
\zeta_{f}(z)=\mathrm{Tr}\left(f(D^{2}+\mu^{2})^{-z/2}\right)\ .
\]
Here $\mu$ is, as usual, a non-zero real number needed to compensate
for the lack of invertibility of $D$. An immediate computation shows
that
\[
\zeta_{f}(z)=\frac{2^{[n/2]}}{(2\pi)^{n}}\int\left(\xi^{2}+\mu^{2}\right)^{-z/2}d^{n}\xi\int f(x)d^{n}x\ ,
\]
where the coefficient $2^{[n/2]}$ comes from the trace over the spinor
bundle. The integral over $\xi$ is finite for $\mathrm{Re}(z)>n$
and we get
\begin{equation}
I_{c}(z):=\int\left(\xi^{2}+\mu^{2}\right)^{-z/2}d^{n}\xi=\pi^{n/2}\mu^{n-z}\frac{\Gamma\left(\frac{z-n}{2}\right)}{\Gamma\left(\frac{z}{2}\right)}\ .\label{eq:comm-int}
\end{equation}
We obtain an analytic continuation using well-known properties of
the gamma function, and we find that the only singularities of $\xi_{f}(z)$
are simple poles. Indeed $\Gamma(z)$ has poles on the negative real
axis at $z=0,-1,-2,\cdots$, so that the function $\Gamma\left(\frac{z-n}{2}\right)$
has poles at $z=n-2m$, where $m\in\mathbb{N}_{0}$. When $n$ is
even the poles at $z=0,-2,-4,\cdots$ are canceled by the zeroes of
$\Gamma\left(\frac{z}{2}\right)$. Then the result is that $\zeta_{f}(z)$
has simple poles at $z=n,n-2,\cdots,2$ when $n$ is even, and has
simple poles at $z=n,n-2,\cdots,1,-1,-3,\cdots$ when $n$ is odd.

For compact Riemannian manifolds this kind of zeta function has been
studied by Minakshisundaram and Pleijel \cite{zeta-compact}, and
here we have the analogous result for $\mathbb{R}^{n}$. We can easily
compute the residue at $z=n$ of $\zeta_{f}(z)$, which is given by
\[
\mathrm{Res}_{z=n}\zeta_{f}(z)=\frac{2^{[n/2]}}{(2\pi)^{n}}\frac{2\pi^{n/2}}{\Gamma\left(\frac{n}{2}\right)}\int f(x)d^{n}x\ .
\]

Now we are ready to study the singularities and the analytic continuation
in the case of $\kappa$-Minkowski space, where the relevant zeta
function is defined by
\[
\zeta_{f}(z):=\Phi\left(f(D^{2}+\mu^{2})^{-z/2}\right)\ ,
\]
where we recall that $\Phi(\cdot)=\mathrm{Tr}(\Delta_{\Phi}\cdot)$
and we omit the representation symbol $\pi$.
\begin{prop}
Let $f\in\mathcal{A}$ and $\mathrm{Re}(z)>n$. Then we have
\[
\zeta_{f}(z)=\frac{2^{[n/2]}}{(2\pi)^{n}}I(z)\int f(x)d^{n}x\ ,
\]
where $I(z)=\frac{1}{2}(I_{c}(z)+I_{\lambda}(z))$, with the function
$I_{c}(z)$ being the classical result given in (\ref{eq:comm-int}),
which is independent of $\lambda$, and the function $I_{\lambda}(z)$
being given by 
\[
I_{\lambda}(z)=\pi^{(n-1)/2}\mu^{(n-1)-z}\frac{\Gamma\left(\frac{z-(n-1)}{2}\right)}{\Gamma\left(\frac{z}{2}\right)}\lambda^{-1}{}_{2}F_{1}\left(\frac{1}{2},\frac{z-(n-1)}{2};\frac{3}{2};-\frac{1}{(\lambda\mu)^{2}}\right)\ .
\]
The function $I(z)$ reduces to the classical one $I_{c}(z)$ in the
limit $\lambda\to0$.\end{prop}
\begin{proof}
From the proof of Proposition \ref{prop:spec-dim} we have
\[
\mathrm{Tr}\left(f\Delta_{\Phi}(D^{2}+\mu^{2})^{-z/2}\right)=\frac{2^{[n/2]}}{(2\pi)^{n}}\int G_{s}^{\Delta}(\xi)d^{n}\xi\int(Uf)(x)d^{n}x\ ,
\]
where we recall that we have set $t=1$. Similarly to the classical
case we set $I(z):=\int G_{s}^{\Delta}(\xi)d^{n}\xi$. We already
partially computed this integral, and the result was
\[
I(z)=\pi^{(n-1)/2}\frac{\Gamma\left(\frac{z-(n-1)}{2}\right)}{\Gamma\left(\frac{z}{2}\right)}\int e^{-t\lambda\xi_{0}}\left(\lambda^{-2}\left(1-e^{-\lambda\xi_{0}}\right)^{2}+\mu^{2}\right)^{-\frac{z-(n-1)}{2}}d\xi_{0}\ .
\]
We need to compute the last integral. First we do the change of variable
$r=e^{-\lambda\xi_{0}}$ and obtain
\[
I(z)=\pi^{(n-1)/2}\frac{\Gamma\left(\frac{z-(n-1)}{2}\right)}{\Gamma\left(\frac{z}{2}\right)}\lambda^{z-n}\int_{0}^{\infty}\left((1-r)^{2}+(\lambda\mu)^{2}\right)^{-\frac{z-(n-1)}{2}}dr\ .
\]
This integral can be computed analytically. We use the formula
\[
\int_{0}^{\infty}\left((1-r)^{2}+a^{2}\right)^{-z}dr=a^{-2z}\left[\frac{a\sqrt{\pi}}{2}\frac{\Gamma\left(z-\frac{1}{2}\right)}{\Gamma(z)}+{}_{2}F_{1}\left(\frac{1}{2},z;\frac{3}{2};-\frac{1}{a^{2}}\right)\right]\ ,
\]
which is valid for $\mathrm{Re}(z)>1/2$. Here $_{2}F_{1}(a,b;c;z)$
is the ordinary hypergeometric function. Therefore the integral in
$I(z)$ is finite for $\mathrm{Re}(z)>n$ and we have
\[
I(z)=\frac{1}{2}\pi^{n/2}\mu^{n-z}\frac{\Gamma\left(\frac{z-n}{2}\right)}{\Gamma\left(\frac{z}{2}\right)}+\frac{1}{2}I_{\lambda}(z)\ ,
\]
where we have defined the function
\[
I_{\lambda}(z):=2\pi^{(n-1)/2}\mu^{(n-1)-z}\frac{\Gamma\left(\frac{z-(n-1)}{2}\right)}{\Gamma\left(\frac{z}{2}\right)}\lambda^{-1}{}_{2}F_{1}\left(\frac{1}{2},\frac{z-(n-1)}{2};\frac{3}{2};-\frac{1}{(\lambda\mu)^{2}}\right)\ .
\]
Notice that we have $I(z)=\frac{1}{2}(I_{c}(z)+I_{\lambda}(z))$.
Finally we have $\int Uf=\int f$, which is valid for $f\in\mathcal{A}$,
from which the first part of the proposition follows.

Now we want to consider the classical limit of $I(z)$, in the case
$\mathrm{Re}(z)>n$. Using the linear transformation formulae for
the hypergeometric function $_{2}F_{1}(a,b;c;z)$ it is easy to obtain
an asymptotic expansion for large negative $z$, see \cite{abra-steg}.
This expansion takes the form
\[
_{2}F_{1}(a,b;c;-z)\sim\frac{\Gamma(c)\Gamma(b-a)}{\Gamma(b)\Gamma(c-a)}z^{-a}+\frac{\Gamma(c)\Gamma(a-b)}{\Gamma(a)\Gamma(c-b)}z^{-b}\ .
\]
With this result it is easy to compute the limit
\[
\lim_{\lambda\to0}\lambda^{-1}{}_{2}F_{1}\left(\frac{1}{2},\frac{z-(n-1)}{2};\frac{3}{2};-\frac{1}{(\lambda\mu)^{2}}\right)=\frac{\sqrt{\pi}}{2}\frac{\Gamma\left(\frac{z-n}{2}\right)}{\Gamma\left(\frac{z-(n-1)}{2}\right)}\mu\ .
\]
Then we see that $I(z)$ reduces to $I_{c}(z)$ in the classical limit
$\lambda\to0$.\end{proof}
\begin{cor}
For $f\in\mathcal{A}$ we have
\[
\mathrm{Res}_{z=n}\zeta_{f}(z)=c_{n}\omega(f)\ ,
\]
where the constant is defined as $c_{n}=\frac{2^{[n/2]}}{(2\pi)^{n}}\frac{\pi^{n/2}}{\Gamma\left(\frac{n}{2}\right)}$.\end{cor}
\begin{proof}
The function $I_{\lambda}(z)$ is regular at $z=n$, so we get one-half
of the classical residue, that is we have $\mathrm{Res}_{z=n}I(z)=\pi^{n/2}/\Gamma\left(\frac{n}{2}\right)$.
The result follows immediately.\end{proof}
\begin{prop}
Let $f\in\mathcal{A}$. Then the zeta function
\[
\zeta_{f}(z)=\frac{2^{[n/2]}}{(2\pi)^{n}}I(z)\int f(x)d^{n}x
\]
has a meromorphic extension to the whole complex plane with only simple
poles.\end{prop}
\begin{proof}
Since $I(z)=\frac{1}{2}(I_{c}(z)+I_{\lambda}(z))$, where $I_{c}(z)$
is the integral (\ref{eq:comm-int}) arising in the commutative case,
the zeta function $\zeta_{f}(z)$ has the poles of the commutative
case plus additional poles coming from the function $I_{\lambda}(z)$.
To study them consider the hypergeometric function $_{2}F_{1}(a,b;c;z)$,
with the assumption that $c$ does not belong to $\{0,-1,-2,\cdots\}$.
The series defining $_{2}F_{1}(a,b;c;z)$ is convergent in the open
disk $|z|<1$, but can be analytically continued to the entire complex
plane with a branch cut from $z=1$ to $z=\infty$. Therefore the
function
\[
_{2}F_{1}\left(\frac{1}{2},\frac{z-(n-1)}{2};\frac{3}{2};-\frac{1}{(\lambda\mu)^{2}}\right)
\]
does not have any poles in $z$. Now recall that the function $I_{\lambda}(z)$
is defined by
\[
I_{\lambda}(z)=2\pi^{(n-1)/2}\lambda^{-1}\mu^{(n-1)-z}\frac{\Gamma\left(\frac{z-(n-1)}{2}\right)}{\Gamma\left(\frac{z}{2}\right)}{}_{2}F_{1}\left(\frac{1}{2},\frac{z-(n-1)}{2};\frac{3}{2};-\frac{1}{(\lambda\mu)^{2}}\right)\ .
\]
Therefore the only poles of this function come from the ratio of the
two gamma functions. These are simple poles, from which the claim
follows.
\end{proof}
It is interesting to note that the poles of the function $I_{c}(z)$
come from the ratio $\Gamma\left(\frac{z-n}{2}\right)/\Gamma\left(\frac{z}{2}\right)$,
while the poles of the function $I_{\lambda}(z)$ come from the ratio
$\Gamma\left(\frac{z-(n-1)}{2}\right)/\Gamma\left(\frac{z}{2}\right)$:
the latter are therefore the poles of the $(n-1)$-dimensional case.
If we think of the Lorentzian version of $\kappa$-Minkowski space,
we can relate this result to the different properties of the time
direction from the space directions, evident already from the commutation
relations.

\section{Twisted homology}

\subsection{Motivation and preliminaries}

In this section we want to study the homological properties of $\kappa$-Minkowski
space. In the non-compact setting it is not completely clear, at least
as far as we understand, which algebra should be considered in this
respect. A possibility is to consider a certain unitalization of the
algebra $\mathcal{A}$ in consideration, as done in \cite{moyal}.
On the other hand, already at the commutative level, if we consider
$\mathbb{R}^{n}$ we would like to have an analog of the volume form
$dx^{1}\wedge\cdots\wedge dx^{n}$, but is clear that the functions
$x^{\mu}$ do not belong to a unital algebra.

Our plan is to investigate the homological properties of the enveloping
algebra $U(\mathfrak{g}_{\kappa})$, where $\mathfrak{g}_{\kappa}$
is the Lie algebra underlying $\kappa$-Minkowski space. General results
for the twisted homology of an enveloping algebra are given in \cite{dual-comp},
where the twist is called the \textit{Nakayama automorphism}. Here
we choose a more elementary approach, which involves the explicit
computation using the Chevalley-Eilenberg complex for the Lie algebra
$\mathfrak{g}_{\kappa}$. This choice also allows us to do a more
detailed comparison with other examples coming from quantum groups.

Let us start by recalling some notions from homological algebra, following
the exposition given in \cite{basic}. Let $\mathfrak{g}$ be a Lie
algebra and $M$ be a left $\mathfrak{g}$-module. The \textit{Lie
algebra homology} of $\mathfrak{g}$ with coefficients in $M$ is,
by definition, the homology of the \textit{Chevalley-Eilenberg complex}

\[
M\xleftarrow{\delta}M\otimes\Lambda^{1}\mathfrak{g}\xleftarrow{\delta}M\otimes\Lambda^{2}\mathfrak{g}\xleftarrow{\delta}\cdots\ ,
\]
where $\Lambda^{k}\mathfrak{g}$ denotes the $k$-th exterior power
of $\mathfrak{g}$ and the differential $\delta$ is defined by
\begin{equation}
\begin{split}\delta(m\otimes X_{1}\wedge\cdots\wedge X_{n}) & =\sum_{i<j}^{n}(-1)^{i+j}m\otimes[X_{i},X_{j}]\wedge X_{1}\wedge\cdots\wedge\hat{X}_{i}\wedge\cdots\wedge\hat{X}_{j}\wedge\cdots\wedge X_{n}\\
 & +\sum_{i=1}^{n}(-1)^{i}X_{i}(m)\otimes X_{1}\wedge\cdots\wedge\hat{X}_{i}\wedge\cdots\wedge X_{n}\ ,
\end{split}
\label{eq:diff}
\end{equation}
where the hat denotes omission. Denote by $U(\mathfrak{g})$ the universal
enveloping algebra of $\mathfrak{g}$. Given a $U(\mathfrak{g})$-bimodule
$M$, we define the left $\mathfrak{g}$-module $M^{ad}$, where $M^{ad}=M$
as vector spaces and the left module structure is defined for all
$X\in\mathfrak{g}$ and $m\in M$ by 
\[
X(m)=Xm-mX\ .
\]
We can define a map
\[
\varepsilon:M^{ad}\otimes\Lambda^{n}\mathfrak{g}\to M\otimes U(\mathfrak{g})^{\otimes n}
\]
from the Lie algebra complex to the Hochschild complex by
\[
\varepsilon(m\otimes X_{1}\wedge\cdots\wedge X_{n})=\sum_{s\in S_{n}}\mathrm{sgn}(s)m\otimes X_{s(1)}\otimes\cdots\otimes X_{s(n)}\ .
\]
One can prove that $\varepsilon:C(\mathfrak{g},M^{ad})\to C(U(\mathfrak{g}),M)$
is a quasi-isomorphism, so it induces an isomorphism between the corresponding
homology groups
\[
H_{*}(\mathfrak{g},M^{ad})\cong H_{*}(U(\mathfrak{g}),M)\ .
\]
In particular if we choose $M={}_{\sigma}U(\mathfrak{g})$, that is
$U(\mathfrak{g})$ with the bimodule structure $a\cdot b\cdot c=\sigma(a)bc$,
then on the right we have the \textit{twisted Hochschild homology}
$H_{*}(U(\mathfrak{g}),{}_{\sigma}U(\mathfrak{g}))$. The \textit{twisted
Hochschild dimension} is defined, according to \cite{dual-comp},
as the maximum of the homological dimension of $H_{*}(U(\mathfrak{g}),{}_{\sigma}U(\mathfrak{g}))$
over all the automorphisms $\sigma$ of $U(\mathfrak{g})$. The case
$\sigma=\mathrm{id}$ gives the usual Hochschild homology. Interest
in this twisted homology theory comes from the fact that, in several
examples coming from quantum groups, it allows to avoid the phenomenon
of dimension drop. We will see that this is the case also here.

\subsection{The two dimensional case}

Now we can start the computation for the two dimensional case. For
clarity we use the notation $x_{1},x_{2}$ instead of $x_{0},x_{1}$
as done in the previous sections. Since the Lie algebra is two dimensional
the complex is simply given by
\[
M\xleftarrow{\delta}M\otimes\Lambda^{1}g\xleftarrow{\delta}M\otimes\Lambda^{2}g\xleftarrow{\delta}0\ .
\]
The differential $\delta$ acting on $M\otimes\Lambda^{2}g$ takes
the form
\[
\delta(m\otimes X_{1}\wedge X_{2})=-m\otimes[X_{1},X_{2}]-X_{1}(m)\otimes X_{2}+X_{2}(m)\otimes X_{1}\ .
\]
We write $X_{1}$ and $X_{2}$ in the $x_{1},x_{2}$ basis as $X_{1}=c_{1}^{1}x_{1}+c_{1}^{2}x_{2}$
and $X_{2}=c_{2}^{1}x_{1}+c_{2}^{2}x_{2}$, for some coefficients.
Their commutator is given by
\[
[X_{1},X_{2}]=(c_{1}^{1}c_{2}^{2}-c_{2}^{1}c_{1}^{2})i\lambda x_{2}\ .
\]
Notice that for $m\otimes X_{1}\wedge X_{2}$ to be non-trivial we
need $c_{1}^{1}c_{2}^{2}-c_{2}^{1}c_{1}^{2}\neq0$. Indeed we have
\[
m\otimes X_{1}\wedge X_{2}=(c_{1}^{1}c_{2}^{2}-c_{2}^{1}c_{1}^{2})m\otimes x_{1}\wedge x_{2}\ .
\]

\begin{prop}
The twisted homological dimension of $U(\mathfrak{g}_{\kappa})$ is
equal to two.\end{prop}
\begin{proof}
Since $\Lambda^{3}\mathfrak{g}$ is trivial we only have to show that
there exists a non-trivial element $m\otimes X_{1}\wedge X_{2}$ such
that $\delta(m\otimes X_{1}\wedge X_{2})=0$. Computing the differential
we get
\[
\delta(m\otimes X_{1}\wedge X_{2})=-(c_{1}^{1}c_{2}^{2}-c_{2}^{1}c_{1}^{2})\left((i\lambda m+x_{1}(m))\otimes x_{2}-x_{2}(m)\otimes x_{1}\right)\ .
\]
Since $c_{1}^{1}c_{2}^{2}-c_{2}^{1}c_{1}^{2}\neq0$, the condition
$\delta(m\otimes X_{1}\wedge X_{2})=0$ implies
\[
(i\lambda m+x_{1}(m))\otimes x_{2}-x_{2}(m)\otimes x_{1}=0\ .
\]
This in turn implies the conditions $x_{2}(m)=0$ and $i\lambda m+x_{1}(m)=0$.

We have $X(m)=\sigma(X)m-mX$, where $\sigma$ is an automorphism
of the form
\[
\sigma(x_{1})=x_{1}+i\mu\ ,\qquad\sigma(x_{2})=x_{2}\ .
\]
By the Poincaré\textendash{}Birkhoff\textendash{}Witt theorem, we
can write $m\in U(\mathfrak{g}_{\kappa})$ as
\[
m=\sum_{a,b}f_{a,b}x_{1}^{a}x_{2}^{b}\ ,
\]
where the sum is finite, $f_{a,b}$ are numerical coefficients and
the exponents are non-negative integers. Since the automorphism $\sigma$
acts trivially on $x_{2}$, the condition $x_{2}(m)=0$ implies that
$m$ commutes with $x_{2}$, that is $m$ should not depend on $x_{1}$.

The second condition, on the other hand, can be rewritten as
\[
i\lambda m+x_{1}(m)=i(\lambda+\mu)m+[x_{1},m]=0\ .
\]
An easy computation then shows that
\[
[x_{1},m]=\sum_{a,b}f_{0,b}[x_{1},x_{2}^{b}]=i\lambda\sum_{a,b}f_{0,b}bx_{2}^{b}\ .
\]
Plugging this result into $i\lambda m+x_{1}(m)=0$ we obtain
\[
\begin{split}i\lambda m+x_{1}(m) & =i(\lambda+\mu)\sum_{a,b}f_{0,b}x_{2}^{b}+i\lambda\sum_{a,b}f_{0,b}bx_{2}^{b}\\
 & =\sum_{a,b}f_{0,b}i(\lambda(1+b)+\mu)x_{2}^{b}=0\ .
\end{split}
\]
Since $b$ is a non-negative integer, this equation is satisfied if
and only if $\mu=-\lambda(1+b)$, for some $b\in\mathbb{N}_{0}$.
In particular, the simplest choice $b=0$ corresponds to the automorphism
$\sigma(x_{1})=x_{1}-i\lambda$, $\sigma(x_{2})=x_{2}$ which has
been considered in \cite{modulark}.
\end{proof}

\subsection{The $n$-dimensional case}

Let us write $\delta=\delta_{1}+\delta_{2}$, where $\delta_{1}$
and $\delta_{2}$ are given respectively by the first and second line
of equation (\ref{eq:diff}). To study the case of a general dimension
we start by proving two lemmata, which allows to rewrite the differential
in a easier form. The first one is valid for any Lie algebra $\mathfrak{g}$,
and simply requires some gymnastics with differential forms, while
the second one is related to the simple structure of the commutation
relations of the Lie algebra $\mathfrak{g}_{\kappa}$.
\begin{lem}
Let $X_{i}\in\mathfrak{g}$ be given by $X_{i}=c_{i}^{j}x_{j}$, where
$c_{i}^{j}$ are numerical coefficients and $\{x_{j}\}$ is a basis
of the Lie algebra $\mathfrak{g}$. Then we have
\[
\delta_{2}(m\otimes X_{1}\wedge\cdots\wedge X_{n})=\det C\sum_{j=1}^{n}x_{j}(m)\otimes x_{1}\wedge\cdots\wedge\widehat{x}_{j}\wedge\cdots\wedge x_{n}\ ,
\]
where $C$ is the matrix formed by the coefficients $c_{i}^{j}$.\end{lem}
\begin{proof}
Denoting by $C_{i,j}$ the $(i,j)$-minor of the matrix $C$ we can
write
\[
X_{1}\wedge\cdots\wedge\widehat{X}_{i}\wedge\cdots\wedge X_{n}=\sum_{j=1}^{n}C_{i,j}x_{1}\wedge\cdots\wedge\widehat{x}_{j}\wedge\cdots\wedge x_{n}\ .
\]
If we expand $X_{i}$ in the basis of the generators we can write
\[
X_{i}(m)=\sum_{k=1}^{n}c_{i}^{k}x_{k}(m)\ .
\]
Then the second line of the differential $\delta$ given by (\ref{eq:diff})
becomes
\[
\delta_{2}=\sum_{j=1}^{n}\sum_{k=1}^{n}x_{k}(m)\otimes\sum_{i=1}^{n}(-1)^{i}c_{i}^{k}C_{i,j}x_{1}\wedge\cdots\wedge\widehat{x}_{j}\wedge\cdots\wedge x_{n}\ .
\]
The sum over $i$ of $(-1)^{i}c_{i}^{k}C_{i,j}$ looks like a Laplace
expansion of the determinant of some matrix. Indeed, it is the determinant
of the matrix obtained from $C$ by replacing the $j$-th column,
given by $c_{a}^{j}$ with $a=1,\cdots,n$, with the column $c_{a}^{k}$.
If $k\neq j$ then, after this replacement, we obviously have two
linearly dependent columns, so the determinant vanishes. On the other
hand if $k=j$ we obtain $\det C$, independent of $j$. So we can
write
\[
\sum_{i=1}^{n}(-1)^{i}c_{i}^{k}C_{i,j}=\det C\delta_{j}^{k}\ .
\]
Plugging this result into the previous formula we find the result.\end{proof}
\begin{lem}
With the same notation as above, consider the Lie algebra $\mathfrak{g}_{\kappa}$
with commutation relations $[x_{1},x_{j}]=i\lambda x_{j}$, where
$j>1$. Then we have
\[
\delta_{1}(m\otimes X_{1}\wedge\cdots\wedge X_{n})=\det Ci\lambda(n-1)m\otimes\widehat{x}_{1}\wedge x_{2}\wedge\cdots\wedge x_{n}\ .
\]
\end{lem}
\begin{proof}
We start by computing the commutator of two elements $X_{i}$
\[
[X_{i},X_{j}]=i\lambda\sum_{k=1}^{n}(c_{i}^{1}c_{j}^{k}-c_{i}^{k}c_{j}^{1})x_{k}=i\lambda(c_{i}^{1}X_{j}-c_{j}^{1}X_{i})\ .
\]
Then the first line of the differential $\delta$ given by (\ref{eq:diff})
becomes
\[
\begin{split}\delta_{1} & =\sum_{i<j}i\lambda c_{i}^{1}(-1)^{i+j}m\otimes X_{j}\wedge X_{1}\wedge\cdots\wedge\hat{X}_{i}\wedge\cdots\wedge\hat{X}_{j}\wedge\cdots\wedge X_{n}\\
 & -\sum_{i<j}i\lambda c_{j}^{1}(-1)^{i+j}m\otimes X_{i}\wedge X_{1}\wedge\cdots\wedge\hat{X}_{i}\wedge\cdots\wedge\hat{X}_{j}\wedge\cdots\wedge X_{n}\ .
\end{split}
\]
Now we can bring $X_{i}$ and $X_{j}$ to their missing spots, picking
up some signs. When we move $X_{i}$ we have to go across $i-1$ terms,
so we pick a $(-1)^{i-1}$, while when we move $X_{j}$ we have to
go across $j-2$ terms, since also $X_{i}$ is missing, so we pick
a $(-1)^{j-2}$. Then we have 
\[
\begin{split}\delta_{1} & =\sum_{i<j}i\lambda c_{i}^{0}(-1)^{i}m\otimes X_{1}\wedge\cdots\wedge\hat{X}_{i}\wedge\cdots\wedge X_{n}\\
 & +\sum_{i<j}i\lambda c_{j}^{0}(-1)^{j}m\otimes\wedge X_{1}\wedge\cdots\wedge\hat{X}_{j}\wedge\cdots\wedge X_{n}\ .
\end{split}
\]
Notice that now $j$ and $i$ do not appear anymore respectively in
the first and the second sum. It is not difficult to see that we can
rewrite them as
\[
\begin{split}\delta_{1} & =\sum_{i=1}^{n}i\lambda(n-i)c_{i}^{1}(-1)^{i}m\otimes X_{1}\wedge\cdots\wedge\hat{X}_{i}\wedge\cdots\wedge X_{n}\\
 & +\sum_{j=1}^{n}i\lambda(j-1)c_{j}^{1}(-1)^{j}m\otimes\wedge X_{1}\wedge\cdots\wedge\hat{X}_{j}\wedge\cdots\wedge X_{n}\ .
\end{split}
\]
Summing the two contributions we get
\[
\delta_{1}=i\lambda(n-1)\sum_{i=1}^{n}(-1)^{i}c_{i}^{1}m\otimes\wedge X_{1}\wedge\cdots\wedge\hat{X}_{i}\wedge\cdots\wedge X_{n}\ .
\]
Writing the wedge products in terms of the minors of $C$ we obtain
\[
\delta_{1}=i\lambda(n-1)\sum_{j=1}^{n}\sum_{i=1}^{n}(-1)^{i}c_{i}^{1}C_{i,j}m\otimes\wedge x_{1}\wedge\cdots\wedge\widehat{x}_{j}\wedge\cdots\wedge x_{n}\ .
\]
Finally, using the same arguments of the previous lemma, we obtain
\[
\delta_{1}=\det Ci\lambda(n-1)m\otimes\widehat{x}_{1}\wedge x_{2}\wedge\cdots\wedge x_{n}\ .
\]
This concludes the proof of the lemma.\end{proof}
\begin{thm}
Let $\mathfrak{g}_{\kappa}$ be the Lie algebra associated with $\kappa$-Minkowski
space in $n$-dimensions, which is characterized by the commutation
relations $[x_{1},x_{j}]=i\lambda x_{j}$, where $j>1$. Then the
twisted homological dimension of $U(\mathfrak{g}_{\kappa})$ is equal
to $n$.\end{thm}
\begin{proof}
As in the two dimensional case, we only need to show that there is
an element $m\otimes X_{1}\wedge\cdots\wedge X_{n}$ such that $\delta(m\otimes X_{1}\wedge\cdots\wedge X_{n})=0$.
Putting together the two previous lemmata we have the following expression
for the differential
\[
\begin{split}\delta(m\otimes X_{1}\wedge\cdots\wedge X_{n}) & =\det Ci\lambda(n-1)m\otimes\widehat{x}_{1}\wedge x_{2}\wedge\cdots\wedge x_{n}\\
 & +\det C\sum_{j=1}^{n}x_{j}(m)\otimes x_{1}\wedge\cdots\wedge\widehat{x}_{j}\wedge\cdots\wedge x_{n}\ .
\end{split}
\]
Since $\det C$ is different from zero, we need to impose the conditions
$x_{j}(m)=0$ for $j=2,\cdots,n$. Again we have to consider automorphisms
which are of the form
\[
\sigma(x_{1})=x_{1}+i\mu\ ,\qquad\sigma(x_{j})=x_{j}\ .
\]
Then $x_{j}(m)=[x_{j},m]=0$ implies that $m$ does not depend on
$x_{1}$. The other condition we need to impose is $i\lambda(n-1)m+x_{1}(m)=0$,
which can be rewritten in the form
\[
i(\lambda(n-1)+\mu)m+[x_{1},m]=0\ .
\]
By the Poincaré\textendash{}Birkhoff\textendash{}Witt theorem we can
write $m\in U(\mathfrak{g}_{\kappa})$ as
\[
m=\sum_{a_{2},\cdots,a_{n}}f_{0,a_{2},\cdots,a_{n}}x_{2}^{a_{2}}\cdots x_{n}^{a_{n}}\ ,
\]
where the sum is finite, $f_{0,a_{2},\cdots,a_{n}}$ are numerical
coefficients and the exponents are non-negative integers. We have
already imposed the condition that $m$ does not depend on $x_{1}$.
Now we compute the commutator of $m$ with $x_{1}$
\[
\begin{split}[x_{1},m] & =\sum_{a_{2},\cdots,a_{n}}f_{0,a_{2},\cdots,a_{n}}[x_{1},x_{2}^{a_{2}}\cdots x_{n}^{a_{n}}]\\
 & =\sum_{a_{2},\cdots,a_{n}}f_{0,a_{2},\cdots,a_{n}}\left([x_{1},x_{2}^{a_{2}}]x_{3}^{a_{3}}\cdots x_{n}^{a_{n}}+\cdots+x_{2}^{a_{2}}\cdots x_{n-1}^{a_{n-1}}[x_{1},x_{n}^{a_{n}}]\right)\\
 & =i\lambda\sum_{a_{2},\cdots,a_{n}}f_{0,a_{2},\cdots,a_{n}}(a_{2}+\cdots+a_{n})x_{2}^{a_{2}}\cdots x_{n}^{a_{n}}\ .
\end{split}
\]
Using this result we finally obtain
\[
i\lambda m+x_{1}(m)=\sum_{a_{2},\cdots,a_{n}}f_{0,a_{2},\cdots,a_{n}}i(\lambda(n-1+a_{2}+\cdots+a_{n})+\mu)x_{2}^{a_{2}}\cdots x_{n}^{a_{n}}=0\ .
\]
Since $a_{2},\cdots,a_{n}$ are non-negative integers, this equation
is satisfied if and only if $\mu=-\lambda(n-1+a_{2}+\cdots+a_{n})$,
for some $a_{2},\cdots,a_{n}\in\mathbb{N}_{0}$.
\end{proof}
We notice from this result that, by choosing the simplest case where
$a_{2}=\cdots=a_{n}=0$, we obtain the automorphism given by $\sigma(x_{1})=x_{1}-i(n-1)\lambda$
and $\sigma(x_{j})=x_{j}$. This is exactly the inverse of the modular
group $\sigma_{i}^{\omega}$ of the weight $\omega$ we introduced
in the first part, which was the starting point for the construction.
Other choices of the $a_{2},\cdots,a_{n}$ coefficients give an automorphism
$\sigma$ which is a negative power of this modular group. This is
exactly the same thing that happens for the twisted homology of $SL_{q}(2)$
\cite{SLq(2)}, for the Podle\'{s} spheres \cite{podles-hom} and
for other examples which come from quantum groups.

There is another feature of this result which is worth mentioning.
Consider again the simplest non-trivial cycle, which we obtain by
setting $f_{0,\cdots,0}=1$ and all other coefficients zero. Passing
from from the Lie algebra complex to the Hochschild complex via
\[
\varepsilon(m\otimes X_{1}\wedge\cdots\wedge X_{n})=\sum_{s\in S_{n}}\mathrm{sgn}(s)m\otimes X_{s(1)}\otimes\cdots\otimes X_{s(n)}\ ,
\]
we see that this cycle corresponds to
\[
c=\det C\sum_{s\in S_{n}}\mathrm{sgn}(s)1\otimes x_{s(1)}\otimes\cdots\otimes x_{s(n)}\ .
\]
We notice that it has the same form as in commutative case. Indeed
the analogy goes further since, as we discussed in the section on
the Dirac operator and the differential calculus, we have that $[D,x^{\mu}]_{\sigma}=-i\Gamma^{\mu}$.
Therefore, if we represent this cycle on the Hilbert space by $a_{0}[D,a_{1}]_{\sigma}\cdots[D,a_{n}]_{\sigma}$,
we get exactly the orientation cycle of the commutative case, which
corresponds to the volume form $dx^{1}\wedge\cdots\wedge dx^{n}$.

\section*{Acknowledgements}

I would like to thank Ludwik D\k{a}browski and Gherardo Piacitelli
for their assistance in the preparation of this paper. I am also grateful
to Ulrich Krähmer and Adam Rennie for helpful discussions and comments.

\end{document}